\documentclass{article}

\usepackage[utf8]{inputenc}
\usepackage{amsmath,amsthm,latexsym,amssymb,amsfonts,color}
\usepackage{bbm}

\usepackage[normalem]{ulem}
\usepackage{xcolor}

\usepackage{hyperref}
\usepackage{orcidlink}

\newcommand{\R}{{\mathbb R}}

\newcommand{\Lie}{{\mathcal L}}

\newcommand{\LAMS}{{\mathbb L}}

\newcommand{\ef}{\hat{e}}

\newcommand{\g}{\overline{g}}

\DeclareMathOperator{\const}{const}
\DeclareMathOperator{\rd}{d}

\newtheorem{df}{Definition}
\newtheorem{lm}{Lemma}
\newtheorem{cor}{Corollary}
\newtheorem{thm}{Theorem}
\newtheorem{prop}{Proposition}

\theoremstyle{remark}
\newtheorem{remark}{Remark}

\title{Geometric origin of intrinsic rigidity for extremal horizons}
\author{Wojciech Kamiński \orcidlink{0000-0003-3707-6087}, 
Adam Szereszewski \orcidlink{0000-0002-3242-8963}
}
\date{%
\small
  {\it   Faculty of Physics, University of Warsaw,\\
ul. Pasteura 5, 02-093 Warsaw, Poland}\\[2ex]%
    \today
}

\begin{document}

\maketitle

\begin{abstract}
Our work clarifies the origin of enhanced symmetry of near horizon geometries. We relate this property to an existence of specific foliation by non-expanding horizons. The existence of this foliation can be shown in any spacetime satisfying null convergence condition. Our proof is geometric and based on analyzing null hypersurfaces in near horizon geometry spacetimes.
\end{abstract}

\section{Introduction}

Extremal Killing horizons differ from non-extremal horizons as their intrinsic geometry is highly constrained by Einstein's field equations. For example, in four dimensional Einstein-Maxwell theory, the constraints force the horizon's induced metric and rotation $1$-form to match those of an extremal Kerr-Newman black hole \cite{Lewandowski:2002dq, Dunajski:2023IntrinsicRigidity, Colling:2024Rigidity}. This fact is based on a profound theorem that these intrinsic data need to be axisymmetric \cite{Dunajski:2023IntrinsicRigidity}.\footnote{Extended to electromagnetic field in \cite{Colling:2024Rigidity}.} This can be viewed as an intrinsic counterpart of the Hawking rigidity theorem \cite{cmp/1103857884, Chrusciel:2008js, Chrusciel:2023onh, Hollands:2006rj, Moncrief:2008mr, Hollands:2008wn} that however is independent of any global properties of spacetime. Recently this intrinsic rigidity was shown to hold in a great generality, in any dimension, under only assumption that the Ricci tensor satisfies \emph{null convergence condition} \cite{Colling:2025Symmetries}.
This condition can be stated as follows:  any double contraction of the Ricci tensor with arbitrary null vector is non-negative. In particular, it is satisfied if the spacetime is a solution to Einstein field equations with matter content satisfying null energy condition. The proof of this remarkable theorem \cite{Colling:2025Symmetries} is based on mysterious identity which is not only hard to verify but also has no obvious geometric interpretation. This identity was first developed in the context of vacuum equation in \cite{Dunajski:2023IntrinsicRigidity} and studied extensively in \cite{Colling:2024QuasiEinstein, kaminski2024:extremehorizonequation}. The main goal of the present paper is to derive intrinsic rigidity using more geometric approach based on theory of non-expanding horizons. Our method clarifies also the origin of the crucial identity and as result provides also some simplifications of the proof from \cite{Colling:2025Symmetries}.

The intrinsic rigidity result can be naturally formulated as a theorem regarding near horizon geometries (NHGs) \cite{Kunduri:2013gce} that satisfy null convergence condition. Near horizon geometries are a special class of Kundt metrics obtained from spacetimes containing extremal Killing horizons via a Geroch-type limiting procedure \cite{Kunduri:2013gce, Bardeen:1999px}. Because this limit preserves the  null convergence condition \cite{Colling:2025Symmetries}, the enhanced symmetry of near horizon geometry translates into symmetry property of the geometric data on the extremal Killing horizon.

One particular form of the NHG metric is 
\begin{align}\label{eq:metric}
g&=2\Gamma(x)\rd u \left(\rd \rho+\frac{1}{2}\rho^2 A(x)\rd u\right)+\nonumber\\&+h_{ij}(x)\left(\rd x^i+\rho K^i(x)\rd u\right)\left(\rd x^j+\rho K^j(x)\rd u\right),
\end{align}
where the spacetime is $M=\R\times S\times \R$ (coordinates $u,x^i,\rho$ respectively) and $S$ is a Riemannian manifold with metric $h$. We assume that $S$ is compact  and connected. Other elements are functions $A$ and $\Gamma$ on $S$ (with assumption $\Gamma>0$) and a vector field $K$ on $S$.  The hypersurface
\begin{equation}\label{eq:H}
{\mathcal H}:=\{(u,x^i,\rho)\in M\colon \rho=0\}    
\end{equation}
is an extremal Killing horizon\footnote{Killing horizon is a null hypersurface such that a given Killing vector field is a non-vanishing normal to the hypersurface. Some authors add additional conditions about connectedness \cite{Mars:2018hag, Chrusciel2012, Chrusciel:2020fql} or inextendibility \cite{Chrusciel2012, Chrusciel:2020fql}. Additionally, \cite{Chrusciel2012, Chrusciel:2020fql} exclude certain extremal horizons in the case when Killing vector field is null.} for a Killing vector field $U:=\frac{\partial}{\partial u}$ (called shift symmetry). There is another isometry by rescaling $(u,x^i,\rho)\rightarrow (e^su,x^i,e^{-s}\rho)$ (called boost symmetry) generated by a Killing vector field $B:=u\frac{\partial}{\partial u}-\rho\frac{\partial}{\partial \rho}$. Both symmetries preserves the horizon ${\mathcal H}$.

It turns out that we can assume some additional conditions on the coordinate system. Namely, we may assume that the divergence of the vector field $K$ vanishes
\begin{equation}\label{eq:div-K}
D_iK^i=0,
\end{equation}
where we denote covariant derivative on $S$ by $D_i$. 
This condition for the coordinate system was employed first in study of Aretakis instability \cite{Lucietti:2007uc}. We explain possibility of imposing \eqref{eq:div-K} in Section \ref{sec:principal}.

Intrinsic rigidity is a remarkable fact about possible form of near horizon geometry \cite{Dunajski:2023IntrinsicRigidity, Colling:2024QuasiEinstein, Colling:2025Symmetries}
\begin{thm}\label{thm:1}
Let $(M,g)$ be an NHG spacetime with $D_iK^i=0$ satisfying null convergence condition (that is $R_{\mu\nu}v^\mu v^\nu\geq 0$ for every null vector $v^\mu$) then 
\begin{enumerate}
\item function $A(x)$ is a constant,
\item $\Lie_Kh=0$, $\Lie_K\Gamma=0$.\footnote{The case of vanishing $K$ is also covered by this theorem. However, the second point trivializes.}
\end{enumerate}
\end{thm}

As a consequence of this fact $g$ is preserved by two more vector fields \cite{Colling:2025Symmetries}\footnote{This was first observed in the near horizon limit of extremal Kerr black hole in \cite{Bardeen:1999px}.}
\begin{equation}
Y:=\frac{1}{2}Au^2\frac{\partial}{\partial u}+(1-A\rho u)\frac{\partial}{\partial \rho}-uK^i\frac{\partial}{\partial x^i},\quad K^i\frac{\partial}{\partial x^i}.
\end{equation}
In the proof in \cite{Colling:2025Symmetries}, some vector field plays a crucial role:
\begin{equation}
\ell_C'=\Gamma\left(\frac{\partial}{\partial u}-\rho K^i\frac{\partial}{\partial x^i}-\frac{1}{2}A\rho^2\frac{\partial}{\partial \rho}\right).
\end{equation}
Our observation is that in the situation of Theorem \ref{thm:1}, this vector field is tangent to a foliation by non-expanding horizons. In fact, acting by a group locally generated by $Y$ we shift extremal Killing horizon ${\mathcal H}$ into a family of extremal Killing horizons foliating the spacetime.\footnote{By results of \cite{Booth:2024MOTSinstability} the Killing vector can be either tangent to extremal Killing horizon with a compact  and connected cross section or everywhere transversal to it. Vector $Y$ is transversal to ${\mathcal H}$. If $A\not=0$ then the vector field $Y$ is incomplete, which explains why the shifted hypersurfaces are not connected even if ${\mathcal H}$ is.} 

Our approach to proof of Theorem \ref{thm:1} is based on the solution to the following question. Every NHG spacetime is an example of a Kundt metric. It is foliated by one family of non-expanding horizons 
\begin{equation}\label{eq:H's-graph}
{\mathcal H}_s':=\{(u,x^i,\rho)\colon u=s=\const\},
\end{equation}
with the normal vector $k=\Gamma^{-1}\frac{\partial}{\partial \rho}$ \cite{Pawlowski:2003ys, Lewandowski:2016HigherDim}. One can ask if there exists another such foliation transversal to this one.

The first observation is the following:

\begin{prop}\label{prop:1}
Consider a NHG spacetime $M$ satisfying $D_iK^i=0$ and a null hypersurface given by a graph
\begin{equation}\label{eq:hypersurface}
\{(u,x^i,\rho)\in M\colon \rho=f(u,x^i),\ u\in I\},
\end{equation}
where $f\in C^\infty(I\times S)$ and $I\subset \R$ is an open interval. Suppose that $f$ is not identically zero and the expansion of this hypersurface is non-positive everywhere (or non-negative everywhere). Then
\begin{enumerate}
\item Expansion vanishes identically,
\item Function $A$ is constant.
\item $f(u,x^i)=\rho_s(u)$ for some $s\in \R$, where 
\begin{equation}\label{eq:rho_s}
\rho_s(u):=\frac{1}{s+\frac{1}{2}Au}.
\end{equation}
\end{enumerate}
\end{prop}

Let us remark that sign of expansion depends on the choice of time orientation on the hypersurface. Here we have two choices, either $u$ growing with time or decreasing with time. Both choices are covered by this proposition. In order to prove this result we will compute directly the expansion. It is given by $\theta={\mathbb L}(\Gamma f)$, where ${\mathbb L}$ is an important stability operator (second order operator on $S$) governing variation of expansion of the horizon ${\mathcal H}$. The appearence of this operator is not surprising. Due to high symmetry of NHG spacetime, the formula for expansion is linear in $f$, so the infinitesimal variation gives an exact result. Function $\Gamma$ is the principal eigenfunction of this operator and the principal eigenvalue is zero. From basic properties of this type of operators if for some constant $u$, expansion is non-negative (non-positive) then it vanishes and $f|_{u=\const}=\const$. Using this property we can prove Proposition \ref{prop:1}. The proof will be given in Section \ref{sec:hyper}.

This fact allows us to focus our attention on ${\mathcal H}_s$ hypersurfaces given by the graphs \eqref{eq:hypersurface} with $f(u,x^i)=\rho_s(u)$ (see \eqref{eq:rho_s}). We can express them as level sets,
\begin{equation}\label{eq:Hs-graph}
{\mathcal H}_s:=\{(u,x^i,\rho)\colon \chi(u,\rho)=s\},\quad s\in \R,
\end{equation}
where we introduced a function
\begin{equation}
\chi=\frac{1}{\rho}-\frac{1}{2}Au.
\end{equation}
Our convention is ${\mathcal H}_{\pm \infty}={\mathcal H}$. 
Furthermore, we introduce a vector field $\ell_C$, normal to the leaves of the foliation\footnote{It is a natural (distinguished)  normal vector field to the leaves of a foliation considered in \cite{Pawlowski:2003ys, Lewandowski:2016HigherDim}.}
\begin{equation}
\ell_C\lrcorner g=\rd \chi.
\end{equation}
We can check that this vector is proportional to $\ell_C'$ if $A=\const$. For $A$ constant and non-zero, every hypersurface ${\mathcal H}_s$, $s\in \R$ has two connected components ($u>-\frac{2s}{A}$ and $u<-\frac{2s}{A}$).

Concerning ${\mathcal H}_s$ hypersurfaces, we have the following geometric classification:

\begin{prop}\label{prop:2}
Let $M$ be NHG spacetime with $D_iK^i=0$. A segment $u_-<u<u_+$ of a hypersurface ${\mathcal H}_s$ for $s\in \R$ is
\begin{enumerate}
\item\label{prop2:1}  a null hypersurface with vanishing expansion if and only if function $A(x)$ is constant,
\item\label{prop2:2}  a totally geodesic null hypersurface
\footnote{Totally geodesic means that the second fundamental form vanishes.}
if and only if
\begin{equation}
\Lie_Kh=0,\quad A=\const,
\end{equation}
\item\label{prop2:3} a non-expanding horizon \footnote{
The definition of a non-expanding horizon varies across the literature.
Here we follow \cite{Ashtekar:2021gkq} in contrast to \cite{Ashtekar:2004cn}. The non-expanding horizon is a totally geodesic null hypersurface such that the Einstein tensor contracted with the normal vector is proportional to this vector. The definition of \cite{Ashtekar:2004cn} adds an additional minor condition that proportionality constant is non-positive. } if and only if
\begin{equation}
\Lie_Kh=0,\quad A=\const,\quad \Lie_K\Gamma=0.
\end{equation}
\end{enumerate}
Moreover, if the last condition holds then  ${\mathcal H}_s$   is an extremal Killing horizon for some Killing vector field.
\end{prop}

Point \ref{prop2:1} is a complementary result to Proposition \ref{prop:1} which shows that $A=\const$ is a necessary condition for ${\mathcal H}_s$ to be hypersurface with vanishing expansion. Proposition \ref{prop:2} is a quite straightforward computation except of the last part where some more involved argument is needed. To prove the last point, we will show that if ${\mathcal H}_s$ is a non-expanding horizon then $\Gamma$ is a principal eigenvector of  certain second order operator on $S$ that is invariant under action of the vector field $K$.  This is sufficient to show invariance of $\Gamma$.  The proof of Proposition \ref{prop:2} will be given in Section \ref{sec:hyper}.

Let us remark that we do not assume neither any energy condition nor any Einstein field equations in this proposition.  The Theorem \ref{thm:1} will follow if we can show that ${\mathcal H}_s$ are non-expanding horizons under assumption of null convergence condition:

\begin{prop}\label{prop:3}
If null convergence condition holds then ${\mathcal H}_s$ are non-expanding horizons for $s\in \R$.
\end{prop}

The proof of Proposition \ref{prop:3} is remarkable simple. Let us consider a null hypersurface ${\mathcal N}$ of the form \eqref{eq:hypersurface}. We assume that its cross section at $u_0$ is a surface of constant $\rho$
\begin{equation}
{\mathcal C}=\{(u,x^i,\rho)\colon u=u_0,\ \rho=\rho_0\}.
\end{equation}
We will show that such hypersurface exists for arbitrary $u_0,\rho_0$ and moreover, the expansion vanishes on ${\mathcal C}$.  The expansion on the part of ${\mathcal N}$ in future of ${\mathcal C}$ is non-positive, whereas expansion on the part of ${\mathcal N}$ in past of ${\mathcal C}$ is non-negative. It follows from Raychadhuri equation if null convergence condition is satisfied. By Proposition \ref{prop:1} the expansion vanishes everywhere and additionally, the hypersurface coincides with a segment of ${\mathcal H}_s$ for some $s$. Now under assumption of null convergence condition, every null hypersurface with vanishing expansion is in fact a non-expanding horizon. As result ${\mathcal H}_s$ is non-expanding horizon and Proposition \ref{prop:2} shows that $A$ is constant, $\Lie_Kh=0$ and $\Lie_K\Gamma=0$. The details of the proof will be given in Section \ref{sec:null-energy}.

Finally, it is possible to obtained an identity equivalent to the formula from \cite{Dunajski:2023IntrinsicRigidity, Colling:2024QuasiEinstein, Colling:2025Symmetries}. This identity is obtained from analyzing the Raychadhuri equation in a specific situation. It allows to apply argument from \cite{Colling:2025Symmetries}, but it has a small advantage over the original formulation that constancy of $A$ follows directly from this identity without invoking a separate argument. In order to prove $\Lie_K\Gamma=0$, we deduce another interesting identity based on the proof of Proposition \ref{prop:2} point~\ref{prop2:3}.

Let us now comment on symmetry inheritence of the matter fields. Suppose that NHG spacetime is a solution to Einstein equations with some matter content satisfying null energy condition. The question is if the matter fields themself inherit symmetries of the spacetime. Our goal is to show how using geometric framework of non-expanding horizons simplifies proof of symmetry inheritence. We will restrict our attention to the Maxwell field:

\begin{prop}\label{prop:enhancement}
Let $M$ be a solution of Einstein-Maxwell theory with a cosmological constant. If $M$ is NHG spacetime then Maxwell field is invariant under symmetries $U,B,Y$ and $K^i\frac{\partial}{\partial x^i}$ of NHG.
\end{prop}

We explain the reason for this property on a simpler example of a minimally coupled scalar field. We start with a simple observation. For every non-expanding horizon in a solution to Einstein scalar field equation, the scalar field is constant along generators of the horizon. Indeed, the Einstein tensor contracted with two normal vectors vanishes thus on the horizon
\begin{equation}
0\hat{=}T_{\mu\nu}\ell^\mu\ell^\nu=\left(\ell^\mu\nabla_\mu \phi\right)^2,
\end{equation}
so $\ell^\mu\nabla_\mu \phi\hat{=}0$ at the horizon.

In the case of solution to Einstein scalar field equations, NHG spacetime satisfies null convergence condition thus we have two foliations ${\mathcal H}_s$ and ${\mathcal H}_s'$ of non-expanding horizons. As result in the whole spacetime
\begin{equation}
\left(\frac{\partial}{\partial u}-\rho K^i\frac{\partial}{\partial x^i}-\frac{1}{2}A\rho^2\frac{\partial}{\partial \rho}\right)\phi=0,\quad \frac{\partial \phi}{\partial \rho}=0
\end{equation}
In particular $\phi$ is $\rho$-independent.  Considering first equation at $\rho=0$ we now show that $\left.\frac{\partial\phi}{\partial u}\right|_{\rho=0}=0$, but from $\rho$ independence it holds everywhere. In this case the first equation shows that $K^i\frac{\partial \phi}{\partial x^i}=0$. These properties are sufficient to show that
$\phi$ is preserved by four Killing vector fields of NHG spacetime. Similar reasoning can be applied to Einstein-Maxwell theory. The proof of Proposition \ref{prop:enhancement} will be given in Section \ref{sec:inheritence}.

Proposition \ref{prop:enhancement} shows that the Maxwell field in NHG is necessarily stationary (invariant under shift symmetry), even if we do not assume this property. On the other hand, the result tell us also something about near horizon data of the Maxwell fields. Maxwell fields admit near horizon limit \cite{Kunduri-Reall:2007vf, Kunduri:2013gce}. By applying Proposition \ref{prop:enhancement} to a near horizon limit of a stationary solution of $\Lambda$-Maxwell-Einstein theory with an extremal horizon, we can show that the near horizon data inherit enhanced symmetry \cite{Colling:2025Symmetries}.

\vskip 3pt\noindent
{\bf Plan of the paper:}
Firstly, we will present necessary background in Section \ref{sec:preliminaries}: Subsection
\ref{sec:NHG} introduces basis in NHG and presents the formulas useful in the main part of our work, Subsection \ref{sec:principal} presents crucial results about principal eigenvectors and Subsection \ref{sec:horizon-theory} sketches necessary theory of quasi-local horizons. In Subsection \ref{sec:nec} we provide a proof that in spacetime satisfying null convergence condition every null hypersurface with vanishing expansion is a non-expanding horizon. Sections \ref{sec:hyper} and \ref{sec:null-energy} constitute the main body of the work. They contain proofs of  Propositions \ref{prop:1}, \ref{prop:2} and \ref{prop:3} and Theorem \ref{thm:1}. The version of Dunajski-Lucietti identity is derived in Section \ref{sec:identity}. The final part, Section \ref{sec:inheritence} is devoted to inheritence of the symmetry by the Maxwell field. 

\vskip 3pt\noindent
{\bf Notation:}
We use standard convention that components of the vectors are labelled by upper indices and of forms by lower indices. The small latin indices will have range $1,\ldots, d-2$. They are raised or lowered using metric $h$.  The greek indices will have range $1,\ldots, d$ and they are reserved for objects on the spacetime. They are raised and lowered using metric $g$. We use Einstein summation convention (summation over repeated indices in the equation, one upper and one lower). If possible, we use index free notation. For example, metric $h$ stands for $h_{ij}\rd x^i\rd x^j$. Symbol $:=$ introduce notation (definition), $\hat{=}$ means equality to holds on a hypersurface, which will be clear from the context and symbol $\lrcorner$ denotes contraction of a form or tensor with a vector field. Additionally, $\nabla$ is the covariant derivative defined by metric $g$ on $M$, and $D$ is covariant derivative defined by metric $h$ on $S$.
The signature of the Lorentzian metric on a Lorentzian manifold (spacetime) is $(-+\cdots+)$.  Hypersurface always means a smooth embedded submanifold of codimension $1$.

\subsection*{Acknowledgements}

We would like to thank Alex Colling,  Maciej Dunajski and James Lucietti for comments and discussions about extremal black holes.  We would like to dedicate this work to the memory of Jurek Lewandowski, who turned our interest to theory of quasi-local horizons.

\section{Preliminaries}\label{sec:preliminaries}

We will first describe necessary ingredients of near horizon geometry \cite{Kunduri:2013gce} , relevant properties of second order operators \cite{andersson2008stability,Dunajski:2023IntrinsicRigidity}, and facts from theory of quasi-local horizons \cite{Ashtekar:1998sp, Ashtekar:2004cn, Ashtekar:2001jb, Lewandowski:2004QuasilocalHigherDim}. 

\subsection{Near horizon geometry spacetimes}\label{sec:NHG}

Near horizon geometry (NHG) is a spacetime $M=\R\times S\times \R$ with a Kundt metric of special form \cite{Kunduri:2013gce}
\begin{equation}\label{eq:NHG-original}
g=2\rd u \left(\rd r+rX_i(x)\rd x^i+\frac{1}{2}r^2 F(x)\rd u^2\right)+h_{ij}(x)\rd x^i\rd x^j,
\end{equation}
where $F(x)$, $X_i(x)$ and $h_{ij}(x)$ are a function, a $1$-form and a Riemannian metric on $S$. We assume that $S$ is compact and connected. This Gaussian null coordinates will not be most useful in our analysis.

Let us change coordinate system replacing $r=\Gamma(x) \rho$ where $\Gamma>0$ is a smooth function on $S$ and $\rho$ is the new variable.
The form of the metric is now \eqref{eq:metric} with
\begin{equation}
K_i=D_i\Gamma+X_i\Gamma,\quad A=F\Gamma-\frac{K_i K^i}{\Gamma}.
\end{equation}
We can introduce a basis of one forms labelled by $+,-,1,\ldots, d-2$
\begin{equation}\label{eq:basis}
\ef^+=\rd u,\quad \ef^-=\rd \rho+\frac{1}{2}A\rho^2\rd u,\quad \ef^i=\rd x^i+\rho K^i\rd u.
\end{equation}
This not an orthonormal basis, however, the metric takes a simple form
\begin{equation}
g=2\Gamma \ef^+\ef^-+h_{ij}\ef^i \ef^j.
\end{equation}
The dual basis of vectors is given by
\begin{equation}
e_+=\frac{\partial}{\partial u}-\rho K^i\frac{\partial}{\partial x^i}-\frac{1}{2}\rho^2 A\frac{\partial}{\partial \rho},\quad e_-=\frac{\partial}{\partial \rho},\quad e_i=\frac{\partial}{\partial x^i},
\end{equation}
and the inverse metric can be written using this basis as
\begin{equation}
g^{-1}=2\Gamma^{-1}e_+e_-+h^{ij}e_ie_j.
\end{equation}
Let us notice that $k=\Gamma^{-1}e_-$ and $\ell_C'=\Gamma e_+$. Moreover, if $A=\const$ then $\ell_C=-\frac{1}{\Gamma\rho^2}e_+$.

In our analysis we will need also a connection:
\begin{enumerate}
\item Covariant derivatives of the vector field $e_+$
\begin{align}
\nabla_{e_+} e_+ &=-\rho\,\Gamma^{-1}\big(\Gamma A+K(\Gamma)\big)e_+ 
                   -\frac{1}{2}\rho^2\Gamma h^{ij}\frac{\partial A}{\partial x^j}e_i,\\
\nabla_{e_-} e_+ &=-\frac{1}{2}h^{ij}\Big(K_j+\frac{\partial \Gamma}{\partial x^j}\Big)e_i,  \\ 
\nabla_{e_i} e_+ &=\frac{1}{2}\Gamma^{-1}\Big(\frac{\partial\Gamma}{\partial x^i} -K_i\Big)e_+ 
                   -\rho\, h^{jk}D_{(j} K_{i)}e_k.
\end{align}
These are covariant derivatives that will be mostly used in our work.
\item Similarly, for the vector field $e_-$
\begin{align}
\nabla_{e_+} e_- &= \rho\, Ae_- - \frac{1}{2}h^{ij}\Big(\frac{\partial\Gamma}{\partial x^j} -K_j\Big)e_i ,\\ \nabla_{e_-} e_- &= 0, \\
\nabla_{e_i} e_- &= \frac{1}{2}\Gamma^{-1}\Big(\frac{\partial \Gamma}{\partial x^i}+K_i\Big)e_-.
\end{align}
In particular, $e_-$ generates affinely parametrized null geodesics.
\item Finally, covariant derivatives of the vector fields $e_i$
\begin{align}
 \nabla_{e_+} e_i &= \frac{1}{2}\Gamma^{-1}\Big(\frac{\partial\Gamma}{\partial x^i} -K_i\Big)e_+
                     +\frac{1}{2}\rho^2\frac{\partial A}{\partial x^i}e_- 
                     -\rho\, h^{jk}\big(D_{[k}K_{i]}+\Gamma_{kil} K^l\big)e_j,\\ 
 \nabla_{e_-} e_i &=\frac{1}{2}\Gamma^{-1}\Big(K_i+\frac{\partial \Gamma}{\partial x^i}\Big)e_- ,\\
 \nabla_{e_i} e_j &= \rho\,\Gamma^{-1}D_{(i}K_{j)}e_- + \Gamma^k\,_{ji}e_k,
\end{align}
where $\Gamma_{ijk}$ are Christoffel symbols for the metric $h=h_{ij}\rd x^i \rd x^j$ on $S$.
\end{enumerate}

\subsection{Principal eigenvectors}\label{sec:principal}

In our work we will need some properties of principal eigenvectors of second order elliptic operators on a compact,  connected manifolds. The following results were proven in \cite{andersson2008stability,Dunajski:2023IntrinsicRigidity}.

\begin{lm}\label{lm: principal}
Let $S$ be a compact,  connected Riemannian manifold with metric $h$ and $X'$ be a one form. Introduce
\begin{equation}
{\mathbb L}'(\psi)=-D^iD_i\psi-D^i(X_i'\psi).
\end{equation}
Then there exists a smooth function $\Gamma'>0$ such that ${\mathbb L}'(\Gamma')=0$. Moreover,
\begin{enumerate}
\item If ${\mathbb L}'(\psi)\geq 0$ (or ${\mathbb L}'(\psi)\leq 0$) for some smooth function $\psi$ on $S$ then $\psi$ is proportional to $\Gamma'$ and ${\mathbb L}'(\psi)= 0$.
\item If a vector field $K'$ satisfies $\Lie_{K'}h=0$ and $\Lie_{K'}X'=0$ then $\Lie_{K'}\Gamma'=0$.
\end{enumerate}
\end{lm}

\begin{proof}
Let us notice that ${\mathbb L}'(\psi)$ is a total divergence, thus
\begin{equation}
\int_S{\mathbb L}'(\psi)\sqrt{\det h}\rd^{d-2} x=0.
\end{equation}
Suppose that ${\mathbb L}'(\psi)\geq 0$. Integral of a non-negative function vanishes if and only if this function is zero, thus we obtain ${\mathbb L}'(\psi)=0$. Similar argument shows that if ${\mathbb L}'(\psi)\leq 0$ then ${\mathbb L}'(\psi)=0$ as well.

There exists a principal eigenvalue $\lambda_o$ and principal eigenfunction $\Gamma'$ such that
\begin{equation}
{\mathbb L}'(\Gamma')=\lambda_o\Gamma',\quad \Gamma'>0,
\end{equation}
and moreover eigenvalue $\lambda_o$ is multiplicity free \cite{andersson2008stability}. If $\lambda_o\not=0$ then ${\mathbb L}'(\Gamma')>0$ or ${\mathbb L}'(\Gamma')<0$, thus from what we proved above we obtain contradiction and $\lambda_o=0$. Moreover, as principal eigenvalue is multiplicity free if ${\mathbb L}'(\psi)=0$ then $\psi$ is proportional to $\Gamma'$. This shows first part of the lemma.

Suppose now that $\Lie_{K'}h=0$ and $\Lie_{K'}X'=0$ then as $K'$ also preserves covariant derivative
\begin{equation}
{\mathbb L}'(\Lie_{K'}\Gamma')=\Lie_{K'}\left({\mathbb L}'(\Gamma')\right)=0.
\end{equation}
This means that $\Lie_{K'}\Gamma'$ is proportional to $\Gamma'$. However, integral of $\Lie_{K'}\Gamma'$ over $S$ vanishes, because $K'$ as a Killing vector field is divergence free. Thus this proportionality constant needs to be zero. We showed that $\Lie_{K'}\Gamma'=0$.
\end{proof}

We will call $\Gamma'$ principal eigenvector of ${\mathbb L}'$. It is unique up to rescaling by a positive constant.

An important role in our analysis will be played by the stability operator ${\mathbb L}$ for the horizon ${\mathcal H}$. It governs variation of the expansion if ${\mathcal H}$ is perturbed \cite{andersson2008stability, Booth-Fairhurst-2008, Mars_2012} (as it will be seen in Section \ref{sec:hyper}). The operator has a simple form\footnote{We remind that $X$ is a form introduced in coordinates \eqref{eq:NHG-original}.}
\begin{equation}\label{eq:stability-L}
{\mathbb L}(\psi):=-D^iD_i\psi-D^i(X_i\psi),
\end{equation}
and it is an example of an operator considered in Lemma \ref{lm: principal}. 
A simple computation gives
\begin{equation}
\LAMS(\Gamma\psi)=-D^i\left(\Gamma D_i\psi+K_i\psi\right).
\end{equation}
We notice that $D_iK^i=-{\mathbb L}(\Gamma)$. By Lemma \ref{lm: principal}, there exists unique up to a constant choice of $\Gamma$ such that vector field $K$ is divergence free \eqref{eq:div-K}. Thus, we can always assume this condition in NHG spacetime when we transform from coordinates \eqref{eq:NHG-original} to coordinates \eqref{eq:metric}.

\subsection{Theory of horizons}\label{sec:horizon-theory}

The main objective of quasi-local theory of horizons is to replace teological notion of event horizon with a local or quasi-local objects \cite{Hajicek1973a, Ashtekar:2004cn, Ashtekar:2001jb}. The program introduces a hierarchy of definitions culminating in a notion of isolated horizon which captures many properties of the Killing horizons \cite{Ashtekar:1998sp}. The definitions of hypersurface without expansion, totally geodesic null hypersurface up to a non-expanding horizon are based only on data induced locally on the null hypersurface. The isolated horizon structure requires additionally a specific choice of null normal vector field.  We will not need this additional structure in our work, we direct interested reader to \cite{Ashtekar:2004cn, Ashtekar:2001jb, Ashtekar:1998sp}.
There are various versions of the definitions of non-expanding horizons. Our approach does not assume any energy conditions nor Einstein's equations and it is purely geometrical (compare for example \cite{Ashtekar:2021gkq}).  

Let ${\mathcal N}$ be a null hypersurface (codimenion $1$ submanifold)  in a spacetime $M$ (dimension $d$ and signature $-+\cdots +$). The restricted metric $\g$ has signature $(0,+\cdots +)$. Let $\bar{\ell}$ will be a non-vanishing null vector field on ${\mathcal N}$.\footnote{It exists if $M$ can be time oriented. Otherwise, one can always consider time-orientable double cover of $M$, so this assumption is not restrictive. Moreover, NHG spacetimes are time-orientable.} and $\ell$ the push-forward of this vector into the spacetime. The vector $\ell$ is normal to ${\mathcal N}$ (normal vectors are tangent for null hypersurfaces). Moreover, $\ell$ generates a congruence of null geodesics on ${\mathcal N}$,
\begin{equation}
    \ell^\mu \nabla_\mu \ell^\nu\hat{=}\kappa^{(\bar{\ell})}\ell^\nu.
\end{equation}
The function $\kappa^{(\bar{\ell})}$ is an acceleration (called surface gravity \cite{Ashtekar:2001jb}). It depends on a choice of $\bar{\ell}$.
The choice of $\bar{\ell}$ is not unique but any two such vector fields differ by rescaling by a smooth function $f$ on ${\mathcal N}$.\footnote{There is a distinguished choice of a densitized null vector field (see \cite{CIAMBELLI20261}), but not the vector field itself.}

We can define the expansion $\theta^{(\bar{\ell})}$ of this congruence \cite{Ashtekar:2001jb}. The expansion depends on the choice of $\bar{\ell}$, but it has a simple transformation rule
\begin{equation}
\theta^{(f\bar{\ell})}=f\theta^{(\bar{\ell})},\quad f\in C^\infty({\mathcal N}),\ f\not=0.
\end{equation}
The condition of vanishing expansion is thus independent of the choice of $\bar{\ell}$. Similarly, condition of expansion being positive (or negative) depends only on the choice of time orientation on ${\mathcal N}$.

A null hypersurface is totally geodesic  if and only if $\Lie_{\bar{\ell}}\g=0$ (vanishing of both expansion and shear). Here similarly
\begin{equation}
\Lie_{f\bar{\ell}}\g=f\Lie_{\bar{\ell}}\g,\quad f\in C^\infty({\mathcal N}).
\end{equation}
We see that condition of vanishing of $\Lie_{\bar{\ell}}\g$ is independent of the choice of $\bar{\ell}$. Raychadhuri equation shows that for a null hypersurface 
\begin{equation}
\Lie_{\bar{\ell}}\g=0\Longrightarrow R_{\mu\nu}\ell^\mu\ell^\nu\hat{=}0,
\end{equation}
for any $\ell$ normal to the null hypersurface (we remind that $\hat{=}$ means equality on ${\mathcal N}$). 

Suppose that ${\mathcal N}$ is totally geodesic. Spacetime covariant derivative  induces then a torsion-free covariant derivative $\bar{\nabla}$ on ${\mathcal N}$. This covariant derivative satisfies condition
\begin{equation}
\bar{\nabla}\g=0,\quad \bar{\nabla}\bar{\ell}=\omega^{(\bar{\ell})}\otimes \bar{\ell},
\end{equation}
where $\omega^{(\bar{\ell})}$ is some $1$-form on ${\mathcal N}$. The form $\omega^{(\bar{\ell})}$ is called rotation $1$-form and it satisfies $\kappa^{(\bar{\ell})}=\bar{\ell}\lrcorner \omega^{(\bar{\ell})}$. The change of $\bar{\ell}$ leads to a bit more complicated transformation law
\begin{equation}
\omega^{(f\bar{\ell})}=\omega^{(\bar{\ell})}+\rd \ln f,\quad f\in C^\infty({\mathcal N}),\ f\not=0.
\end{equation}
This means that a two form $\Omega^{(\bar{\ell})}=\rd \omega^{(\bar{\ell})}$ is independent of the choice of $\bar{\ell}$.

It turns out that $\Omega$ encodes important information about Ricci tensor. Denote by ${\mathfrak f}^{(\bar{\ell})}$ a pull-back to the null hypersurface of a one form  $R_{\mu\nu}\ell^\nu\rd x^\mu$.
The following equality holds on a totally geodesic null hypersurface \cite{Ashtekar:2001jb, Lewandowski:2004QuasilocalHigherDim}
\begin{equation}\label{eq:non-expanding-cond}
{\mathfrak f}^{(\bar{\ell})}=\bar{\ell}\lrcorner \Omega,
\end{equation}
for any $\ell$ normal to the null hypersurface.

\begin{df}
A totally geodesic null hypersurface is called non-expanding horizon if  $\bar{\ell}\lrcorner \Omega=0$.
\end{df}

Suppose that a spacetime is foliated by totally geodesic null hypersurfaces. Normal null vectors are defined only up to rescaling by a function, but in the case of foliation there exists a distinguished choice of normals unique up to rescaling by a function constant on every leaf of foliation \cite{Pawlowski:2003ys, Lewandowski:2016HigherDim}.\footnote{The natural choice of normal defines on leaves of foliation a  distinguished structure of so called extremal weakly isolated horizons.} For any function $\xi$ defining foliation, we introduce $\ell^\mu=\nabla^\mu \xi$. Choosing another $\xi$ rescales $\ell$ but only by a function constant on every leaf.  Moreover as $\ell$ is a gradient
\begin{equation}
\ell^\mu \nabla_\mu \ell^\nu=\ell^\mu \nabla^\nu \ell_\mu=\frac{1}{2}\nabla^\nu (\ell^\mu\ell_\mu)=0.
\end{equation}
In particular, denoting corresponding vector on leaf of foliation by $\bar{\ell}$, 
\begin{equation}
\bar{\ell}\lrcorner \omega^{(\bar{\ell})}=0.
\end{equation}
As the rescaling by a constant function on a hypersurface does not change rotation form, $\omega^{(\bar{\ell})}$ is uniquely associated with the foliation. Cartan formula shows for any function $f$
\begin{equation}
\Lie_{f\bar{\ell}}\omega^{(\bar{\ell})}=f\bar{\ell}\lrcorner \Omega+\rd\left(f\bar{\ell}\lrcorner \omega^{(\bar{\ell})}\right)=f\bar{\ell}\lrcorner \Omega.
\end{equation}
We conclude that a leaf of a foliation by totally geodesic null hypersurfaces is a non-expanding horizon if and only if on this leaf $\Lie_{f\bar{\ell}}\omega^{(\bar{\ell})}=0$ for any non-vanishing $f$.

Let us notice that our definitions does not assume apart of time orientability of $M$ any topological properties of hypersurfaces. 
We will now describe the only topological restriction on null hypersurfaces that will play a role in our analysis.

The null hypersurfaces under consideration will be diffeomorphic to a trivial $\R$ bundle with the null vector field tangent to the fibers. Moreover, the base manifold of this bundle will be compact and connected. In this situation, a cross section of ${\mathcal N}$ will mean a cross section of this bundle.

Finally, we remind notion of a Killing horizon \cite{Mars:2018hag}.

\begin{df}
A null hypersurface ${\mathcal N}$ is a Killing horizon for a Killing vector field $V$ of $g$ if $V$ restricted to ${\mathcal N}$ is a non-vanishing normal vector field to ${\mathcal N}$. It is called extremal if the acceleration of this vector field vanishes ($\kappa^{(\bar{V})}=0$).
\end{df}

\begin{remark}
The same null hypersurface can be Killing horizon with respect to two different Killing vector fields \cite{Mars:2018hag}. For example, a segment $u>0$ of ${\mathcal H}$ \eqref{eq:H} is an extremal Killing horizon for $U$ and a non-extremal Killing horizon for $B$. Note that our definition does not require that the horizon is invariant under the finite isometries generated by the Killing vector field as in \cite{Racz1992}. In fact, many Killing vector fields will be incomplete in NCG spacetimes.
\end{remark}

Let us notice that as $V$ is a Killing vector field, its restriction $\bar{V}$ to ${\mathcal N}$ preserves the induced metric $\g$. As result every Killing horizon is a totally geodesic null hypersurface. Moreover, because $V$ preserves covariant derivative, $\bar{V}$ preserves the induced connection $\bar{\nabla}$. Since $\bar{V}$ also commutes with itself, the  definition of $\omega^{(\bar{V})}$ ensures that
\begin{equation}
\Lie_{\bar{V}}\omega^{(\bar{V})}=0.
\end{equation}
Cartan formula shows that ${\mathcal N}$ is a non-expanding horizon if and only if $\rd \kappa^{(\bar{V})}=0$. Indeed,
\begin{equation}
\bar{V}\lrcorner \Omega=\Lie_{\bar{V}}\omega^{(\bar{V})}-\rd\left(\bar{V}\lrcorner\omega^{(\bar{V})}\right)=-\rd \kappa^{(\bar{V})}.
\end{equation}
We see that a Killing horizon is a non-expanding horizon if and only if the surface gravity $\kappa^{(\bar{V})}$ is constant on its connected components. In particular, extremal Killing horizons are always non-expanding horizons.

\subsection{Null convergence condition}\label{sec:nec}

We will now describe condition imposed on Ricci tensor.

\begin{df}
A spacetime $M$ satisfies null convergence condition (NCC) if in every point of $M$, for every null vector $v$
\begin{equation}
R_{\mu\nu}v^\mu v^\nu\geq 0.
\end{equation}
\end{df}

Our definition does not assume that Einstein equation (with some matter fields) holds, but if this happens for the matter satisfying null energy condition, then the spacetime satisfies null convergence condition as well.
We will show now that if null convergence condition is satisfied then vanishing expansion is sufficient for the null hypersurface to be a non-expanding horizon. 
This fact can be easily deduced from results of \cite{Ashtekar:2004cn, Lewandowski:2004QuasilocalHigherDim, Hounnonkpe2025}.

\begin{prop}\label{prop:non-expanding}
Suppose that null convergence condition holds in $M$. Then every null hypersurface with vanishing expansion in $M$ is a non-expanding horizon.
\end{prop}

\begin{proof}
Consider a null hypersurface ${\mathcal N}$ in $M$ with a null non-vanishing vector field $\bar{\ell}$. The first part of the argument uses Raychadhuri equation that can be written as (see \cite{Ashtekar:2000hw, Lewandowski:2004QuasilocalHigherDim})
\begin{equation}
\Lie_{\bar{\ell}}\theta^{(\bar{\ell})}=\kappa^{(\bar{\ell})}\theta^{(\bar{\ell})}-\frac{1}{d-2}(\theta^{(\bar{\ell})})^2-\left|\sigma^{(\bar{\ell})}\right|^2-R_{\mu\nu}\ell^\mu\ell^\nu.
\end{equation}
Here, $\sigma^{(\bar{\ell})}$ is the shear tensor (defined on the screen bundle) and $\left|\sigma^{(\bar{\ell})}\right|^2\geq 0$ is its squared norm. The norm vanishes if only if shear vanishes. The function $\kappa^{(\bar{\ell})}$ is an acceleration of $\ell$. Let us remind that $\Lie_{\bar{\ell}}\g=0$ if and only if both expansion and shear are zero.

If expansion vanishes then Raychadhuri equation simplifies
\begin{equation}
\left|\sigma^{(\bar{\ell})}\right|^2+R_{\mu\nu}\ell^\mu\ell^\nu\hat{=}0.
\end{equation}
If null convergence condition is satisfied then, it is a sum of two non-negative terms. Consequently both shear vanishes and $R_{\mu\nu}\ell^\mu\ell^\nu\hat{=}0$ (see \cite{Ashtekar:2004cn, Lewandowski:2004QuasilocalHigherDim}). This proves that under assumption of null convergence condition every null hypersurface with vanishing expansion is totally geodesic.

It remained to show that there exists a function $c\in C^\infty({\mathcal N})$ such that
\begin{equation}
R_{\mu\nu}\ell^\nu\hat{=}c\ell_\mu,
\end{equation}
that is equivalent to condition ${\mathfrak f}^{(\bar{\ell})}=0$ 
(see \eqref{eq:non-expanding-cond}). This will follow from an algebraic fact (compare Proposition 1 in \cite{Hounnonkpe2025}):

\begin{lm}\label{lm:null}
Let $W$ be a real finite dimensional vector space and $\alpha_\pm\colon W\times W\rightarrow \R$ bilinear symmetric forms. Assume that 
\begin{enumerate}
\item $\alpha_-$ is non-degenerate,
\item $\alpha_+(v,v)\geq 0$ for every $v\in W$ such that $\alpha_-(v,v)=0$.
\end{enumerate}
Let $v_0\in W$ be such that $\alpha_+(v_0,v_0)=\alpha_-(v_0,v_0)=0$. Then there exists $c\in \R$ such that for every $v\in W$
\begin{equation}
\alpha_+(v_0,v)=c\alpha_-(v_0,v).
\end{equation}
\end{lm}

\begin{proof}
We can assume $v_0\not=0$ (the case of $v_0=0$ is straightforward).
Consider
\begin{equation}
{\mathcal V}=\{v\in W\colon \alpha_-(v,v)=0,\ v\not=0\}\subset W.
\end{equation}
Form $\alpha_-$ is non-degenerate thus the differential of function $\alpha_-(v,v)$ is nonvanishing for $v\not=0$ and as result ${\mathcal V}$ is a submanifold of codimension $1$. Smooth function $\alpha_+(v,v)$ on ${\mathcal V}$ is non-negative thus $v_0$ is a critical point on ${\mathcal V}$ (a minimum). By Lagrange multipier method there exists $c\in \R$ such that $v_0$ is a critical point on $W$ of a function
\begin{equation}
f(v)=\alpha_+(v,v)-c\alpha_-(v,v).
\end{equation}
The derivative of $f$ at $v_0$ in direction $v$ is given by a simple formula:
\begin{equation}
\partial_vf(v_0)=2\left(\alpha_+(v_0,v)-c\alpha_-(v_0,v)\right).
\end{equation}
As the derivative vanishes at $v_0$ we obtain the result.
\end{proof}

We can apply this fact choosing 
\begin{equation}
\alpha_+(v,v)=R_{\mu\nu}v^\mu v^\nu,\quad \alpha_-(v,v)=g_{\mu\nu}v^\mu v^\nu.
\end{equation}
Together with $v_0=\ell$, they satisfy assumptions of Lemma \ref{lm:null} at every point of ${\mathcal N}$, so there exists function $c$ such that $R_{\mu\nu}\ell^\nu\hat{=}c\ell_\mu$. This function is necessarily smooth because $\ell$ is non-vanishing.
\end{proof}

Let us state a simple corollary from Proposition \ref{prop:non-expanding}.

\begin{cor}
Let $M$ be a spacetime satisfying null convergence condition. Then, every Killing horizon has constant surface gravity on its connected components.
\end{cor}

\begin{proof}
Every Killing horizon is a totally geodesic null hypersurface.
In the spacetime satisfying null convergence condition, every totally geodesic null hypersurface is a non-expanding horizon by Proposition \ref{prop:non-expanding}. However, a Killing horizon is a non-expanding horizon if and only if its surface gravity is constant on its connected components.
\end{proof}

The zeroth law of thermodynamics holds under null convergence condition which is weaker than usually assumed in this context dominant energy condition \cite{Bardeen1973, Chrusciel2012},

\section{Null hypersurfaces}\label{sec:hyper}

In this section we will prove Proposition \ref{prop:1} and Proposition \ref{prop:2}. 

We will now consider null hypersurfaces given by the graph,
\begin{equation}\label{eq:N-graph}
{\mathcal N}=\{(u,x,\rho)\in M\colon \rho=f(u,x),\ u\in I\},
\end{equation}
in a NHG spacetime $M$. Here $I$ is an open interval of $\R$ and $f\in C^\infty(I\times S)$ is a smooth function. We parametrize ${\mathcal N}$ by $u,x^i$ (local coordinate system). 
Hypersurface ${\mathcal N}$ has compact cross sections of constant $u$
\begin{equation}
\Psi_{u_o}\colon S\rightarrow {\mathcal N},\quad \Psi_{u_o}(x^i)=(u_o,x^i,f(u_o,x^i)).
\end{equation}
Let us also introduce set of forms on ${\mathcal N}$ defined by 
\begin{equation}\label{eq:e-K-bar}
\bar{e}^i:=\rd x^i +\tilde{K}^i\rd u,\quad \tilde{K}^i:=fK^i+\Gamma D^if,
\end{equation}
and a vector field $\bar{\ell}_o$ on ${\mathcal N}$
\begin{equation}
\bar{\ell}_o:=\frac{\partial}{\partial u}-\tilde{K}^i\frac{\partial}{\partial x^i}.
\end{equation}
Let us remind that $D_if=\frac{\partial f}{\partial x^i}$ and we raise and lower Latin indices using metric~$h$. We will regard $\tilde{K}$ as $u$-dependent vector field on $S$.

\begin{lm}
Condition for ${\mathcal N}$ given by \eqref{eq:N-graph} to be null hypersurface is equivalent to
\begin{equation}\label{eq:f-u}
\frac{\partial f}{\partial u}=-\frac{1}{2}Af^2+ f K^iD_if+\frac{1}{2}\Gamma D_ifD^if.
\end{equation}
Moreover, if this holds then metric $\g$ on ${\mathcal N}$ is given by
\begin{equation}\label{eq:frak-g}
\g=h_{ij}(x)\bar{e}^i\bar{e}^j,
\end{equation}
and the null direction is spanned by $\bar{\ell}_o$.
\end{lm}

\begin{proof}
The condition for ${\mathcal N}$ being null is equivalent to vanishing norm of $\eta=\rd \rho-\rd f$. We can write this form in the basis \eqref{eq:basis}
\begin{align}
\eta&=\left(\rd \rho+\frac{1}{2}A\rho^2\rd u\right)-\left(\frac{\partial f}{\partial u}-\rho K^i\frac{\partial f}{\partial x^i}+\frac{1}{2}A\rho^2\right)\rd u-\frac{\partial f}{\partial x^i}\left(\rd x^i+\rho K^i\rd u\right)=\nonumber\\
&=\ef^--\left(\frac{\partial f}{\partial u}-\rho K^i\frac{\partial f}{\partial x^i}+\frac{1}{2}A\rho^2\right)\ef^+-\frac{\partial f}{\partial x^i}\ef^i.
\end{align}
Using knowledge of scalar product and the fact that $\rho=f$ on ${\mathcal N}$ we obtain equation
\begin{equation}
-2\Gamma^{-1}\left(\frac{\partial f}{\partial u}-f K^i\frac{\partial f}{\partial x^i}+\frac{1}{2}Af^2\right)+h^{ij}\frac{\partial f}{\partial x^i}\frac{\partial f}{\partial x^j}=0,
\end{equation}
that can be written in a nice form \eqref{eq:f-u}.

The induced metric can be easily computed using $\rd \rho=\rd f$ and the formula  \eqref{eq:f-u} giving \eqref{eq:frak-g}. Vector $\bar{\ell}_o=\frac{\partial}{\partial u}-\left(fK^i+\Gamma h^{ij}\frac{\partial f}{\partial x^j}\right)\frac{\partial}{\partial x^i}$ annihilates all forms $\bar{e}^i$ thus it is the null vector for the metric $\g$.
\end{proof}

Consider a covariant tensor $F$ on a null hypersurface ${\mathcal N}$ \eqref{eq:N-graph} whose arbitrary contraction with $\bar{\ell}_o$ vanishes. We call such tensors semi-basic. As forms $\bar{e}^i$ spanned a basis of forms annihilated by $\bar{\ell}_o$, it follows that $F$ can be uniquely written as
\begin{equation}
F=F_{i_1\cdots i_k}(u,x)\bar{e}^{i_1}\otimes\cdots \otimes\bar{e}^{i_k}.
\end{equation}
Moreover $\Psi_u^*(F)=F_{i_1\cdots i_k}(u,x)\rd x^{i_1}\otimes\cdots \otimes\rd x^{i_k}$. In order to determine $F$ it is enough to know its all cross sections of constant $u$. 

Let us notice that $\Lie_{\bar{\ell}_o}F$ is also semi-basic tensor. It is quite straighforward to compute its Lie derivative.

\begin{lm}\label{lm:semi-basic}
Let $F$ be semi-basic covariant tensor on a null hypersurface ${\mathcal N}$ given by equation \eqref{eq:N-graph}. Then
\begin{equation}
\Psi_u^*\left(\Lie_{\bar{\ell}_o}F\right)=\frac{\partial \bar{F}}{\partial u}-\Lie_{\tilde{K}}\bar{F},\quad \bar{F}=\Psi_u^*F.
\end{equation}
Here we regard $\bar{F}$ as $u$ dependent tensor on $S$. 
\end{lm}

\begin{proof}
We will use identity
\begin{equation}
\Lie_{\bar{\ell}_o}F=\Lie_{\frac{\partial}{\partial u}}F-\Lie_{v}F,\quad v=\tilde{K}^i\frac{\partial}{\partial x^i}.
\end{equation}
Let us notice that $\Lie_{\frac{\partial}{\partial u}}\bar{e}^i$ is proportional to $\rd u$, thus
\begin{equation}
\Psi_u^*\left(\Lie_{\frac{\partial}{\partial u}}F\right)=\frac{\partial \bar{F}}{\partial u}.
\end{equation}
On the other hand $v$ is tangent to cross sections of constant $u$ and it is a push-forward of the vector field $\tilde{K}$, hence
\begin{equation}
\Psi_u^*\left(\Lie_{v}F\right)=\Lie_{\tilde{K}}\Psi_u^*F=\Lie_{\tilde{K}}\bar{F}.
\end{equation}
This finishes the proof.
\end{proof}

We can now apply Lemma \ref{lm:semi-basic} to the degenerate metric $\g$ itself. 

\begin{lm}\label{lm:Pij}
Suppose that ${\mathcal N}$ is given by \eqref{eq:N-graph} and \eqref{eq:f-u} is satisfied, then
$\Lie_{\bar{\ell}_o}\g=2P_{ij}\bar{e}^i\bar{e}^j$ where
\begin{equation}\label{eq:Pij}
P_{ij}:=-D_{(i}\tilde{K}_{j)},
\end{equation}
where $\tilde{K}$ is regarded as $u$-dependent vector field on $S$.
\end{lm}

\begin{proof}
Tensor $\Lie_{\bar{\ell}_o}\g$ is annihilated by $\bar{\ell}_o$ thus it can be written in a form $\Lie_{\bar{\ell}_o}\g=2P_{ij}\bar{e}^i\bar{e}^j$ for some $u$ dependent symmetric tensor $P$. 
Lemma \ref{lm:semi-basic} shows using $\Psi_u^*\g=h$ that
\begin{equation}
2P_{ij}=\frac{\partial h_{ij}}{\partial u}-\Lie_{\tilde{K}}h_{ij}=-2D_{(i}\tilde{K}_{j)}.
\end{equation}
This shows \eqref{eq:Pij}.
\end{proof}

\begin{proof}[Proof of Proposition \ref{prop:1} and Point \ref{prop2:1} of Proposition \ref{prop:2}]
The expansion with respect to the null vector field $\bar{\ell}_o$ is given by
\begin{equation}
\theta=h^{ij}P_{ij}=-(D^i(\Gamma D_if)+D^i(K_if)).
\end{equation}
It can be written as $\theta={\mathbb L}(\Gamma f)$ where $\LAMS$ is the stability operator for ${\mathcal H}$ \eqref{eq:stability-L}.
Let us remind that $\Gamma$ is a principal eigenvector of this operator if  coordinate system satisfies condition $D_i K^i=0$.

In particular Lemma \ref{lm: principal} shows that if for a given $u_o$
\begin{equation}
\theta|_{u=u_o}={\mathbb L}(\Gamma f|_{u=u_o})\geq 0,
\end{equation}
then $\theta|_{u=u_o}={\mathbb L}(\Gamma f|_{u=u_o})=0$ and moreover $\Gamma f|_{u=u_o}$ is proportional to $\Gamma$ thus  $f|_{u=u_o}$ is a constant. The same holds if $\theta|_{u=u_o}\leq 0$. Under assumption in Proposition \ref{prop:1} expansion vanishes identically and  we can write $f(u,x^i)=\phi(u)$ for some function $\phi$.

Equation \eqref{eq:f-u} simplifies considerably:
\begin{equation}
\phi'(u)=-\frac{1}{2}A\phi(u)^2.
\end{equation}
We assumed that $\phi(u)$ is not identically zero, thus
solutions exists if and only if $A$ is a constant and then
\begin{equation}
\frac{1}{\phi(u)}=\frac{1}{2}Au+s,
\end{equation}
where $s$ is an arbitrary constant. We conclude that $\phi=\rho_s$ (defined by \eqref{eq:rho_s}) and the hypersurface is a segment of ${\mathcal H}_s$ defined in \eqref{eq:Hs-graph}. 
\end{proof}

\begin{remark}
The formula for expansion $\theta|_{u=u_o}={\mathbb L}(\Gamma f|_{u=u_o})$ is linear in $f$. It is a necessary consequence of boost and shift symmetry and can be proven directly without computing actual form of the operator. From linearity it follows that ${\mathbb L}(\Gamma f|_{u=u_o})$ is in fact an infinitesimal variation of the expansion of the Killing horizon ${\mathcal H}$ under variation $f\frac{\partial}{\partial \rho}$. This explains occurence of the same operator ${\mathbb L}$ in analysis of this variation.
\end{remark}

\begin{proof}[Proof of Proposition \ref{prop:2} point \ref{prop2:2}]
Using results of Proposition \ref{prop:2} point \ref{prop2:1}, we can assume that expansion vanishes and $A=\const$. We need to show that condition $\Lie_{\bar{\ell}_o}\g=0$ is equivalent to $K$ being a Killing vector. Vanishing of $\Lie_{\bar{\ell}_o}\g$ is equivalent to vanishing of $P_{ij}$. Using $f(u,x^i)=\rho_s(u)$ for some $s$,
\begin{equation}
P_{ij}=-D_{(i}(\rho(u)K_{j)})=-\rho_s(u)D_{(i}K_{j)}.
\end{equation} 
Function $\rho_s(u)$ is non-vanishing, thus $K$ is a Killing vector field if and only if $\Lie_{\bar{\ell}_o}\g=0$. 
\end{proof}

In order to analyze further properties of ${\mathcal H}_s$ we need to know the form $\omega^{(\bar{\ell}_C)}$ that is the natural rotation form associated to foliation by these hypersurfaces.
Suppose that $A=\const$ and $K$ is a Killing vector field of metric $h$ (necessary and sufficient conditions for ${\mathcal H}_s$ to be totally geodesic null hypersurface). The form $\omega^{(\bar{\ell}_C)}$ is semi-basic ($\bar{\ell}_C\lrcorner \omega^{(\bar{\ell}_C)}=0$) thus
\begin{equation}
\omega^{(\bar{\ell}_C)}=w_i\bar{e}^i,
\end{equation}
for some $u$ dependent form $w=w_i\rd x^i$ on $S$. This form can be computed from a definition of rotation form by the formula $\nabla_{e_i}\ell_C=w_i\ell_C$. Indeed, using  $\ell_C=-\frac{1}{\Gamma\rho^2}e_+$,
\begin{equation}
\nabla_{e_i}\ell_C=-\frac{1}{2\Gamma}\left(K_i+D_i\Gamma\right)\ell_C.
\end{equation}
Let us also remark that in this case $\bar{\ell}_o$ is proportional to $\bar{\ell}_C$.

\begin{lm}\label{lm:omega-1}
Suppose that $A=\const$ and $K$ is a Killing vector field of metric $h$ then the rotation $2$-form $\Omega$ on the totally geodesic null hypersurface ${\mathcal H}_s$ satisfies $\bar{\ell}_o\lrcorner \Omega=F_i\bar{e}^i$ where $u$ dependent form $F=F_i\rd x^i$ on $S$ is given by
\begin{equation}
F=\rho_s(u) \Lie_K\left[\frac{1}{2\Gamma}\left(K_i+D_i\Gamma\right)\rd x^i\right].
\end{equation}
\end{lm}

\begin{proof}
Let us remind that $\omega^{(\bar{\ell}_C)}$ is semi-basic, hence Lemma \ref{lm:semi-basic} shows that $\Lie_{\bar{\ell}_o}\omega^{(\bar{\ell}_C)}=F_i\bar{e}^i$ for some $u$ dependent form $F$ on $S$. Moreover,  
\begin{equation}
w=\Psi_u^*\left(\omega^{(\bar{\ell}_C)}\right)=-\frac{1}{2\Gamma}\left(K_i+D_i\Gamma\right)\rd x^i,
\end{equation}
and using fact that $\tilde{K}=\rho_s(u)K$, Lemma \ref{lm:semi-basic} shows that
\begin{equation}
F=\frac{\partial w}{\partial u}-\Lie_{\tilde{K}}w=\rho_s(u) \Lie_K\left[\frac{1}{2\Gamma}\left(K_i+D_i\Gamma\right)\rd x^i\right].
\end{equation}
As a leaf of foliation $\bar{\ell}_o\lrcorner \Omega=\Lie_{\bar{\ell}_o}\omega^{(\bar{\ell}_C)}$ that concludes the proof.
\end{proof}

\begin{proof}[Proof of Proposition \ref{prop:2} point \ref{prop2:3}]
Suppose that a segment of ${\mathcal H}_s$ is a non-expanding horizon. Lemma \ref{lm:omega-1} shows that a form $Z=Z_i\rd x^i$, where $Z_i=\frac{1}{\Gamma}\left(K_i+D_i\Gamma\right)$ satisfies $\Lie_KZ=0$. Consider an operator
\begin{equation}
{\mathbb L}'(\psi)=-D^iD_i\psi+D^i(Z_i\psi).
\end{equation}
It satisfies ${\mathbb L}'(\Gamma)=0$ so by Lemma \ref{lm: principal}, $\Gamma$ is a principal eigenvector. Moreover, $\Lie_Kh=0=\Lie_KZ$, so again by Lemma \ref{lm: principal}, $\Lie_K\Gamma=0$. 

We now prove the implication in the other direction. Suppose that additionally to $A=\const$ and $\Lie_Kh=0$ also $\Lie_K\Gamma=0$ then as $K$ commutes with itself 
we can conclude that $\Lie_KZ=0$. By Lemma  
\ref{lm:omega-1} ${\mathcal H}_s$ is a non-expanding horizon for every $s\in \R$.
\end{proof}

We can now finish the proof of Proposition \ref{prop:2}. We will not use local isometries as elucidated in the Introduction, but instead we will explicitely construct a Killing vector field for which given ${\mathcal H}_s$ is an extremal Killing horizon. The last part of Proposition \ref{prop:2} follows directly from the lemma:

\begin{lm}\label{lm:prop2-last}
Suppose that NHG spacetime $M$ satisfies $A=\const$, $\Lie_Kh=0$ and $\Lie_K\Gamma=0$. Then ${\mathcal H}_s$ for $s\in \R$ is an extremal Killing horizon for the Killing vector field
\begin{equation}
V_{(s)}=\frac{1}{2}AY+AsB+s^2 U-sK^i\frac{\partial}{\partial x^i}.
\end{equation}
\end{lm}

\begin{remark}
If $A=0$ and $s=0$ then $V_{(0)}=0$ as well, but in this situation ${\mathcal H}_s$ is an empty set. 
\end{remark}

\begin{proof}
On the hypersurface ${\mathcal H}_s$, $\frac{1}{\rho}=\frac{1}{2}Au+s$ thus
\begin{equation}
\frac{1}{\rho^2}e_+
\hat{=}\left(\frac{1}{2}Au+s\right)^2\frac{\partial}{\partial u}-\left(\frac{1}{2}Au+s\right)K^i\frac{\partial}{\partial x^i}-\frac{1}{2}A\frac{\partial}{\partial \rho}.
\end{equation}
On the other hand, direct computation allows to express right hand side using Killing vector fields $U,B$ and $Y$
\begin{align}
&\left(\frac{1}{2}Au+s\right)^2\frac{\partial}{\partial u}-\left(\frac{1}{2}Au+s\right)K^i\frac{\partial}{\partial x^i}-\frac{1}{2}A\frac{\partial}{\partial \rho}=\nonumber\\
&=\frac{1}{2}AY+AsB+s^2 U-sK^i\frac{\partial}{\partial x^i}-A\rho\left(\frac{1}{\rho}-\frac{1}{2}Au-s\right)\frac{\partial}{\partial \rho}.
\end{align}
As result $\frac{1}{\rho^2}e_+\hat{=}V_{(s)}$ and the null hypersurface ${\mathcal H}_s$ is a Killing horizon for the Killing vector field $V_{(s)}$. In order to keep our notation consistent, we introduce a null vector field $\bar{V}_{(s)}$ on ${\mathcal H}_s$, which push-forward to $V_{(s)}$. In order to compute $\kappa^{(\bar{V}_{(s)})}$ we notice that $V_{(s)}\hat{=}-\Gamma \ell_C$,
\begin{equation}
V_{(s)}^\mu \nabla_\mu V_{(s)}^\nu\hat{=}\Gamma \ell_C^\mu\nabla_\mu\left(\Gamma \ell_C^\nu\right)=\Gamma \left(\ell_C^\mu\nabla_\mu \Gamma\right)\ell_c^\nu,
\end{equation}
where we used fact that as natural normal for the leaf of foliation $\ell_C^\mu\nabla_\mu\ell_C^\nu=0$. However,
\begin{equation}
\ell_C^\mu\nabla_\mu \Gamma=-\frac{1}{\rho^2}\left(\frac{\partial}{\partial u}-\rho K^i\frac{\partial}{\partial x^i}-\frac{1}{2}A\rho^2\frac{\partial}{\partial \rho}\right)\Gamma=0,
\end{equation}
because $\Lie_K\Gamma=0$. Consequently, the acceleration vanishes $\kappa^{(\bar{V}_{(s)})}=0$ and the Killing horizon is extremal.
\end{proof}

Let us notice that if $A\not=0$ then $V_{(s)}$ is not complete vector field. 
This can be remedied by considering an extension of the NHG spacetime (see \cite{Colling:2025Symmetries}) in which the vector fields $U,B,Y$ and $K^i\frac{\partial}{\partial x^i}$ can be integrated into a group action.

Finally,  if $A=\const$ and $\Lie_Kh=0=\Lie_K\Gamma$ then one can check that on ${\mathcal H}_s'$ (see (\ref{eq:H's-graph}))
\begin{equation}
e_-\hat{=}Y-AsB+sK^i\frac{\partial}{\partial x^i}+\frac{1}{2}As^2U.
\end{equation}
As $\nabla_{e_-}e_-=0$ it follows that ${\mathcal H}_s'$ are also extremal Killing horizons for $s\in \R$.

\section{Null convergence condition}
\label{sec:null-energy}

We will start with a simple observation:

\begin{lm}\label{lm:hypersurface}
Let $M$ be arbitrary NHG, $u_0\in \R$ and $f_0\in C^\infty(S)$. Then there exists an open interval $I$ including $u_0$ and a function $f\in C^\infty(I\times S)$ such that
\begin{enumerate}
\item $f(u_0,x)=f_0(x)$ for every $x\in S$
\item Hypersurface \eqref{eq:N-graph} ${\mathcal N}=\{(u,x,\rho)\in M\colon \rho=f(u,x),\ u\in I\}$ is null.
\end{enumerate}
\end{lm}

\begin{proof}
The existence of $f$ satisfying \eqref{eq:f-u} follows from existence of solution to first order PDE by the method of characteristics with initial data on $\{u_0\}\times S$ being $f_0$. In order to prove that this surface is non-characteristic we need to compute $\ell\lrcorner \rd u$ where $\ell\lrcorner g=\rd \rho-\rd f$. However, as $k\lrcorner g=\rd u$ we obtain
\begin{equation}
\ell\lrcorner \rd u=k\lrcorner (\rd \rho-\rd f)=\Gamma^{-1}\frac{\partial}{\partial \rho}\lrcorner (\rd \rho-\rd f)=\Gamma^{-1}\not=0,
\end{equation}
hence the surface of constant $u$ is not-characteristic and the solution exists for some open interval $I$.
\end{proof}

We can now prove Proposition \ref{prop:3}. Suppose that null convergence condition is satisfied and $D_iK^i=0$.
Let $u_0,\rho_0$ be such that $\rho_0\not=0$. Consider ${\mathcal N}$ from Lemma \ref{lm:hypersurface} with $f_0(x)=\rho_0$. Lemma \ref{lm:Pij} shows that the expansion $\theta=h^{ij}P_{ij}$ at $u_0$ vanishes. By Raychadhuri equation and null convergence condition the expansion is non-increasing thus ${\mathcal N}_+={\mathcal N}\cap \{u>u_0\}$ is a null hypersurface with non-positive expansion. However, from Proposition \ref{prop:1},  the expansion vanishes identically on ${\mathcal N}_+$ and moreover ${\mathcal N}_+$ coincides with a segment of ${\mathcal H}_s$. Null hypersurfaces with vanishing expansion are non-expanding horizons in the spacetimes satisfying null convergence condition (Proposition \ref{prop:non-expanding}). Proposition \ref{prop:2} shows now
\begin{equation}
\Lie_Kh=0,\quad A=\const,\quad \Lie_K\Gamma=0.
\end{equation}
This finishes the proof of Proposition \ref{prop:3} and consequently also of Theorem \ref{thm:1}.

\section{Dunajski-Lucietti identity}\label{sec:identity}

We will now derive a version of Dunajski-Lucietti identity for a NHG spacetime satisfying null convergence condition. Our identity is very similar to the one obtained by Colling \cite{Colling:2025Symmetries} but differs slightly. In particular, our identity allows to directly show that $A$ is constant. Let us first introduce two objects on $S$
\begin{equation}
\gamma=R_{\mu\nu}{e_+}^\mu{e_+}^\nu|_{u=0,\rho=1},\quad \beta_i=R_{\mu\nu}{e_i}^\mu {e_+}^\nu|_{u=0,\rho=1}.
\end{equation}
Function $\gamma$ was introduced in \cite{Colling:2025Symmetries}. If null convergence condition holds then $\gamma\geq 0$ as vector field $e_+$ is null. Moreover, Lemma \ref{lm:null} shows that under this assumptions if $\gamma=0$ then $\beta=\beta_i\rd x^i$ vanishes.

\begin{prop}\label{prop:Ray}
Let $M$ be NHG spacetime then
\begin{equation}
\frac{1}{2}\LAMS(\Gamma A)+\left(A+\Lie_K\ln \Gamma\right) D_iK^i-\Lie_K(D_iK^i)=D_{(i}K_{j)}D^{(i}K^{j)}+\gamma.
\end{equation}
Moreover, if $A=\const$ and $\Lie_Kh=0$ then
\begin{equation}
\Lie_KZ=2\beta,
\end{equation}
where $Z=\frac{1}{\Gamma}\left(K_i+D_i\Gamma\right)\rd x^i$.
\end{prop}

Before proving this proposition we will show how Theorem \ref{thm:1} follows from these identities. 
Suppose that $D_iK^i=0$ then the first formula from Proposition \ref{prop:Ray} simplifies
\begin{equation}
\frac{1}{2}\LAMS(\Gamma A)=D_{(i}K_{j)}D^{(i}K^{j)}+\gamma.
\end{equation}
The left hand side is a total divergent. If null convergence condition is satisfied then $\gamma\geq 0$. The integral of total divergence is zero, but the right hand side is non-negative thus both terms on the right hand side need to vanish,
\begin{equation}
\gamma=0,\quad D_{(i}K_{j)}D^{(i}K^{j)}=0.
\end{equation}
The result is that $K$ is the Killing vector. Moreover, now identity simplifies even further, $\frac{1}{2}\LAMS(\Gamma A)=0$. By Lemma \ref{lm: principal}, $A=\const$. We can now use the second formula from Proposition \ref{prop:Ray}:
\begin{equation}
\Lie_KZ=2\beta=0,
\end{equation}
where we used a fact that $\beta=0$ due to vanishing $\gamma$. We prove $\Lie_K\Gamma=0$ by argument from Section \ref{sec:hyper}. The operator $\LAMS'(\psi)=-D_iD^i\psi+D^i(Z_i\psi)$ is invariant under vector field $K$ and $\Gamma$ is its principal eigenvector ($\LAMS'(\Gamma)=0$).  By Lemma \ref{lm: principal}, $\Lie_K\Gamma=0$.

\begin{proof}[Proof of Proposition \ref{prop:Ray}]
Let us consider an arbitrary NHG (without assumption $D_i K^i=0$). Our goal is to analyze Raychadhuri equation for a null hypersurface ${\mathcal N}$ given by $\rho=f(u,x)$. 
Let us remind that $\Psi_u^*\g=h$ and $\Psi_u^*(\Lie_{\bar{\ell}_o}\g)=2P$ where
\begin{equation}
P_{ij}=-D_{(i}\tilde{K}_{j)}.
\end{equation}
The Raychadhuri equation for non-affinely parametrized geodesics takes the form\footnote{The term $P_{ij}P^{ij}$ combines together contributions from both trace part (expansion) and trace-free part of $P_{ij}$ (the shear).}
\begin{equation}\label{eq: Raychadhuri-1}
\Psi_{u}^*(\Lie_{\bar{\ell}_o}\theta)=\kappa h^{ij}P_{ij}-P_{ij}P^{ij}-\Psi_{u}^*(R_{\mu\nu}\ell_o^\mu \ell_o^\nu),
\end{equation}
where $\ell_o$ is a an image under embedding of $\bar{\ell}_o$ into the spacetime and $\kappa$ is the accelaration of $\ell_o$ and $\theta=h^{ij}P_{ij}$ is the expansion. 

Let us remind that $h^{ij}P_{ij}={\mathbb L}(\Gamma f)$. Moreover,
\begin{equation}
\bar{\ell}_o=\frac{\partial}{\partial u}-\tilde{K}^i\frac{\partial}{\partial x^i}.
\end{equation}
This allows us to compute
\begin{equation}
\Psi_{u}^*(\Lie_{\bar{\ell}_o}\theta)=\left(\frac{\partial}{\partial u}-\tilde{K}^i\frac{\partial}{\partial x^i}\right){\mathbb L}(\Gamma f)={\mathbb L}\left(\Gamma \frac{\partial f}{\partial u}\right)-\Lie_{\tilde{K}}{\mathbb L}(\Gamma f).
\end{equation}
We can also compute a push-forward $\ell_o$ of the vector field $\bar{\ell}_o$
\begin{equation}
    \ell_o=\frac{\partial}{\partial u}-\tilde{K}^i\frac{\partial}{\partial x^i}+\left(\frac{\partial f}{\partial u}-\tilde{K}^i\frac{\partial f}{\partial x^i}\right)\frac{\partial}{\partial \rho}=e_+-\Gamma D^if e_i-\frac{1}{2}\Gamma D_if D^if e_-,
\end{equation}
where we used fact that $\rho=f$ on the null hypersurface.

Suppose now that at $u=u_o$ function $f$ is constant and equal $1$ then
\begin{equation}
\left.\frac{\partial f}{\partial u}\right|_{u=u_o}=-\frac{1}{2}A,\quad \tilde{K}|_{u=u_o}=K,\quad P_{ij}|_{u=u_o}=-D_{(i}K_{j)}.
\end{equation}
Moreover, the $\ell_o$ vector field simplifies,
\begin{equation}
\ell_o|_{u=u_o}=e_+.
\end{equation}
We can now compute acceleration $\kappa$. It is defined by the formula $\ell_o^\mu\nabla_\mu \ell_o^\nu=\kappa \ell_o^\nu$ on the hypersurface ${\mathcal N}$. 
Firstly, let us notice that $k_\mu \ell_o^\mu=1$ thus
\begin{equation}
\kappa=k_\nu \ell_o^\mu\nabla_\mu \ell_o^\nu=\ell_o^\mu\nabla_\mu (k_\nu\ell_o^\nu)-{\ell_o}_\nu \ell_o^\mu\nabla_\mu k^\nu=-{\ell_o}_\nu \ell_o^\mu\nabla_\mu k^\nu.
\end{equation}
Moreover, $k=\Gamma^{-1}e_-$ so $\nabla_\mu k^\nu=\Gamma^{-1}\nabla_\mu {e_-}^\nu-\Gamma^{-2}(\nabla_\mu \Gamma) {e_-}^\nu$. This allows us to write
\begin{equation}
\kappa=-\Gamma^{-1}g\left(\ell_o,\nabla_{\ell_o}e_-\right)+\Gamma^{-1}\nabla_{\ell_o}\Gamma.
\end{equation}
We use the formulas for covariant derivatives of $e_-$ and the form of the vector $e_+$ to compute at $\rho=1$
\begin{equation}
\nabla_{e_+}\Gamma=-\Lie_K\Gamma,\quad g\left(e_+,\nabla_{e_+}e_-\right)=\Gamma A.
\end{equation}
At $u=u_o$ we use $\ell_o|_{u=u_o}=e_+$ to obtain
\begin{equation}
\kappa|_{u=u_o}=-\left(A+\Lie_K\ln \Gamma\right).
\end{equation}
Additionally, using the value of $\ell_o$ vector field, $\Psi_{u_o}^*(R_{\mu\nu}\ell_o^\mu \ell_o^\nu)=\gamma$.

We can now write our identity evaluating Raychadhuri equation at $u=u_o$
\begin{equation}
-\frac{1}{2}\LAMS(\Gamma A)+\Lie_K(D_iK^i)=\left(A+\Lie_K\ln \Gamma\right) D_iK^i-D_{(i}K_{j)}D^{(i}K^{j)}-\gamma.
\end{equation}
This shows first part of Proposition \ref{prop:Ray}.

By Lemma \ref{lm:omega-1}, $\frac{1}{2}\Lie_KZ$ is equal to $\bar{\ell}_o\lrcorner\Omega$ restricted to a cross section of constant $u=u_o$. However, on the totally geodesic null hypersurface $\bar{\ell}_o\lrcorner\Omega$ is equal to restriction of $R_{\mu\nu}\ell_o^\nu\rd x^\mu$. Restricting it further to the cross section of constant $u=u_o$ we obtain exactly the form $\beta$. This finishes the proof of the second part. 
\end{proof}

\section{Symmetry inheritence}\label{sec:inheritence}

We will now prove Proposition \ref{prop:enhancement}. It shows that enhanced  symmetries of NHG are inherited by eletromagnetic field.

Solution of $\Lambda$-Einstein-Maxwell theory consists of a metric $g$ and a closed two-form $F$ (Maxwell field). The field $F$ with components $F_{\mu\nu}$ satisfies Maxwell equations 
\begin{equation}
\nabla^\mu F_{\mu\nu}=0,\quad \rd F=0.    
\end{equation}
The Einstein field equations for this matter content are given by
\begin{equation}
G_{\mu\nu}+\Lambda g_{\mu\nu}=\varkappa T_{\mu\nu},\quad T_{\mu\nu} = F_{\mu\alpha} F_{\nu}^{\phantom{\mu}\alpha} - \frac{1}{4} g_{\mu\nu} F_{\alpha\beta} F^{\alpha\beta},
\end{equation}
where $G_{\mu\nu}=R_{\mu\nu}-\frac{1}{2}R^\alpha_{\phantom{\alpha}\alpha} g_{\mu\nu}$ is the Einstein tensor and $\varkappa>0$ is a constant. The cosmological constant $\Lambda$ is arbitrary but fixed.

Let us remind that every solution to $\Lambda$-Einstein-Maxwell theory satisfies null convergence condition. Indeed, if $v$ is a null vector then Einstein's field equations show
\begin{equation}
    R_{\mu\nu}v^\mu v^\nu=\varkappa F_{\mu\alpha} F_{\nu}^{\phantom{\mu}\alpha}v^\mu v^\nu=\varkappa m^\mu m_\mu,
\end{equation}
where $m_\nu=F_{\mu\nu}v^\mu$. The vector $m$ is orthogonal to $v$ because $F$ is a two-form. Every vector orthogonal to a null vector is either null or spacelike. In particular, $m_\mu m^\mu\geq 0$. This shows that NCC condition is satisfied (the theory satisfies null energy condition).

We start  our analysis with an important property of non-expanding horizons in the solutions to $\Lambda$-Einstein-Maxwell theory.

\begin{lm}\label{lm:Maxwell}
Let $M$ be a solution to $\Lambda$-Einstein-Maxwell theory. Suppose that ${\mathcal N}$ is a non-expanding horizon in $M$ with a non-vanishing null vector field $\bar{\ell}$ (and a corresponding push-forwarded normal vector $\ell$). Then:
\begin{enumerate}
\item There exists a smooth function $c\in C^\infty({\mathcal N})$ such that $F_{\mu\nu}\ell^\nu\hat{=}c\ell_\nu$. Moreover, $\Lie_{\bar{\ell}}c=0$.
\item The pull back $\bar{F}$ of $F$ to ${\mathcal N}$ satisfies $\bar{\ell}\lrcorner \bar{F}=0$ (it is semi-basic on ${\mathcal N}$) and $\Lie_{\bar{\ell}}\bar{F}=0$.
\end{enumerate}
\end{lm}

\begin{proof}
Let us notice that for null vectors
\begin{equation}
T_{\mu\nu} \ell^\mu \ell^\nu=m_\mu m^\mu,\quad m_\mu=F_{\mu\nu}\ell^\nu.
\end{equation}
On the non-expanding horizon $T_{\mu\nu} \ell^\mu \ell^\nu\hat{=}0$ thus $m_\mu$ is null. Moreover, $m_\mu \ell^\mu=0$ and the only option for a null vector being orthogonal to $\ell$ is that $m$ is proportional to it. There exists a smooth function $c$ on ${\mathcal N}$ such that $F_{\mu\nu}\ell^\nu\hat{=}c\ell_\mu$. 

We extend $c$ to a function $\tilde{c}$ on $M$ and $\ell$ to a vector field $\tilde{\ell}$ on $M$ such that $\tilde{\ell}^\mu\tilde{\ell}_\mu=0$. This is always possible. We also introduce a function $\xi$ such that $\xi=0$ on ${\mathcal N}$ but $\rd \xi\not=0$ on ${\mathcal N}$. We notice that, the gradient of this function is proportional to $\ell$ on ${\mathcal N}$. Let us consider
\begin{equation}
f_\mu=F_{\mu\nu}\tilde{\ell}^\nu-\tilde{c}\tilde{\ell}_\mu.
\end{equation}
It is perpendicular to $\tilde{\ell}$, namely
\begin{equation}\label{eq:f-ell}
f_\mu \tilde{\ell}^\mu=F_{\mu\nu}\tilde{\ell}^\nu\tilde{\ell}^\mu-\tilde{c}\tilde{\ell}_\mu\tilde{\ell}^\mu=0.
\end{equation}
As $f_\mu$ vanishes at the non-expanding horizon, it can be written around ${\mathcal N}$ in a form $f_\mu=\xi s_\mu$ for some smooth form $s_\mu$. From \eqref{eq:f-ell}, $\xi s_\mu \tilde{\ell}^\mu=0$ and by continuity $s_\mu \tilde{\ell}^\mu\hat{=}0$. In particular $s_\mu\nabla^\mu\xi \hat{=}0$.

Moreover,
\begin{equation}
\nabla^\mu f_\mu=(\nabla^\mu\xi)s_\mu+\xi\nabla^\mu s_\mu\hat{=}0.
\end{equation}
We can compute this divergence also using fact that $\nabla^\mu F_{\mu\nu}=0$,
\begin{equation}\label{eq:div-f}
\nabla^\mu f_\mu=(\nabla^\mu F_{\mu\nu})\tilde{\ell}^\nu+F_{\mu\nu}\nabla^\mu\tilde{\ell}^\nu-\tilde{\ell}^\nu \nabla_\nu \tilde{c}-\tilde{c}\nabla^\nu\tilde{\ell}_\nu=F_{\mu\nu}\nabla^\mu\tilde{\ell}^\nu-\tilde{\ell}^\nu \nabla_\nu \tilde{c}-\tilde{c}\nabla^\nu\tilde{\ell}_\nu.
\end{equation}
On the non-expanding horizon pull-back of $\nabla_\mu\tilde{\ell}_\nu$ to the horizon vanishes, thus we can write
\begin{equation}
\nabla_\mu\tilde{\ell}_\nu\hat{=}u_\mu \tilde{\ell}_\nu+\tilde{\ell}_\mu v_\nu,
\end{equation}
for some vector fields $u$ and $v$. Moreover, as $\tilde{\ell}^\nu\tilde{\ell}_\nu=0$,
\begin{equation}
\tilde{\ell}^\nu\nabla_\mu\tilde{\ell}_\nu=0\Longrightarrow v_\nu\tilde{\ell}^\nu\hat{=}0.
\end{equation}
Using this fact we compute
\begin{equation}
F_{\mu\nu}\nabla^\mu\tilde{\ell}^\nu\hat{=}F_{\mu\nu}(\tilde{\ell}^\mu v^\nu+u^\mu\tilde{\ell}^\nu)=cu^\nu\tilde{\ell}_\nu,\quad \nabla^\mu\tilde{\ell}_\mu\hat{=}u^\mu\tilde{\ell}_\mu+v^\mu\tilde{\ell}_\mu\hat{=}u^\mu\tilde{\ell}_\mu.
\end{equation}
The formula \eqref{eq:div-f} for $\nabla^\mu f_\mu$ simplifies on ${\mathcal N}$
\begin{equation}
\nabla^\mu f_\mu\hat{=} -\ell^\nu \nabla_\nu \tilde{c}.
\end{equation}
As it vanishes on ${\mathcal N}$, we obtain $\Lie_{\bar{\ell}}c=0$.

Pull-back of equality $F_{\mu\nu}\ell^\mu\hat{=}c\ell_\nu$ to ${\mathcal N}$ shows
\begin{equation}
    \bar{\ell}\lrcorner \bar{F}=0.
\end{equation}
The second point of the lemma follows from the Cartan formula
\begin{equation}
    \Lie_{\bar{\ell}}\bar{F}=\bar{\ell}\lrcorner \rd \bar{F}+\rd\left(\bar{\ell}\lrcorner \bar{F}\right)=0.
\end{equation}
We used fact that  $\rd\bar{F}=0$ (due to functoriality of exterior derivative $\rd$).
\end{proof}

Let us now consider an NHG spacetime $M$ that is a solution to Einstein-Maxwell equations with cosmological constant $\Lambda$. We choose coordinate system such that $D_iK^i=0$. 
Because the Einstein-Maxwell theory satisfies the null energy condition, $M$ satisfies the null convergence condition. Thus, we have two foliations by non-expanding horizons ${\mathcal H}_s$ with normal vectors $e_+$ and ${\mathcal H}_s'$ with normal vectors $e_-$.

The most general form of the Maxwell field is
\begin{equation}
F=a\ef^+\wedge \ef^-+ a_i^+\ef^i\wedge \ef^++a_i^-\ef^i\wedge \ef^-+\frac{1}{2}f_{ij}\ef^i\wedge \ef^j.
\end{equation}
The coefficients $a$, $a_i^\pm$ and $f_{ij}$ are functions on $M$.

Consider now a non-expanding horizon ${\mathcal H}_s'$. We have
\begin{equation}
    e_-\lrcorner F=-\left(a\ef^++a_i^-\ef^i\right),\quad e_-\lrcorner g=\Gamma \ef^+,
\end{equation}
and the pull-back $\tilde{f}'$  of $F$ to ${\mathcal H}_s'$ is equal to
\begin{equation}
 \tilde{f}'=a_i^-\rd x^i\wedge \rd \rho+\frac{1}{2}f_{ij}\rd x^i\wedge \rd x^j.  
\end{equation}
We parametrize ${\mathcal H}_s'$ by coordinates $\rho,x ^i$ and then the null vector field is $\frac{\partial}{\partial \rho}$. Lemma
\ref{lm:Maxwell} shows that
\begin{equation}
    a_i^-=0,\quad \Lie_{\frac{\partial}{\partial \rho}}\left(\Gamma^{-1} a\right)=0,\quad \Lie_{\frac{\partial}{\partial \rho}}\tilde{f}'=0.
\end{equation}
The last two conditions simply mean that $a$ and $f_{ij}$ are $\rho$-independent. These conditions hold in the whole spacetime because $s$ is arbitrary.

We will now consider a non-expanding horizon ${\mathcal H}_s$. We compute
\begin{equation}
e_+\lrcorner F=a\ef^--a_i^+\ef^i,\quad e_+\lrcorner g=\Gamma \ef^-.
\end{equation}
Similarly, as in the case of ${\mathcal H}_s'$ non-expanding horizons, we see that $a_i^+=0$ (Lemma \ref{lm:Maxwell}). Moreover, as pull-back of $\ef^i$ to ${\mathcal H}_s$ is equal to $\bar{e}^i$, the pull-back of $F$ to ${\mathcal H}_s$ is equal to $\tilde{f}=\frac{1}{2}f_{ij}\bar{e}^i\wedge \bar{e}^j$. The null vector field on ${\mathcal H}_s$ is given by $\bar{\ell}_o=\frac{\partial}{\partial u}-\rho(u)K^i\frac{\partial}{\partial x^i}$.
Lemma \ref{lm:Maxwell} shows
\begin{equation}
    a^+_i=0,\quad \Lie_{\bar{\ell}_o}\left(\Gamma^{-1} a\right)=0,\quad \Lie_{\bar{\ell}_o}\tilde{f}=0.
\end{equation}
Again this holds on every leaf of foliation.

We can now apply Lemma \ref{lm:semi-basic} to conclude 
\begin{equation}
    \Gamma^{-1} \frac{\partial a}{\partial u}-\rho(u)\Lie_K\left(\Gamma^{-1} a\right)=0,\quad \frac{\partial \bar{f}}{\partial u}-\rho(u)\Lie_K\bar{f}=0.
\end{equation}
where we regard $a$ as $u$-dependent function on $S$ and the form $\bar{f}=\frac{1}{2}f_{ij}\rd x^i\wedge \rd x^j$ is a $u$-dependent form on $S$.
These formulas holds for arbitrary $s$ and as $\bar{f}$ and $a$ are $\rho$-indepedent we obtain very restrictive conditions. Consider first case $s=\pm \infty$, then $\rho(u)=0$. This shows that $a$ and $\bar{f}$ are $u$-independent
\begin{equation}
    \frac{\partial a}{\partial u}=0,\quad \frac{\partial \bar{f}}{\partial u}=0.
\end{equation}
Consider now an arbitrary other $s$ then using fact that $\Gamma$ is Lie dragged by $K$ and $\rho(u)\not=0$
\begin{equation}\label{eq:K-inv}
    \Lie_Ka=0,\quad \Lie_K\bar{f}=0.
\end{equation}
In order to analyze invariance of $F$ under four symmetries of NHG spacetime we write it in the holonomic basis using $a_i^\pm=0$,
\begin{equation}
    F=a\rd u\wedge \rd \rho+\frac{1}{2}f_{ij}\rd x^i\wedge \rd x^j+\rho f_{ij}K^i\rd u\wedge \rd x^j.
\end{equation}
The invariance under boost and shift follows from $u$ and $\rho$-independence of $a$ and $f_{ij}$. The invariance under vector field $K^i\frac{\partial}{\partial x^i}$ follows from \eqref{eq:K-inv}. In order to show invariance under $Y$
we need an additional identity. We compute
\begin{equation}
    0=\rd F=\rd a\wedge \rd u\wedge \rd \rho+\rd \bar{f}-f_{ij}K^i\rd x^j\wedge \rd u\wedge \rd \rho-\rho \rd\left(f_{ij}K^i\rd x^j\right)\wedge \rd u,
\end{equation}
where abusing notation $\rd a$ (respectively $\rd \bar{f}$) is a pull-back of $\rd a$ (respectively $\rd f$) from $S$ by a natural projection along $u$ and $\rho$ variables.
At $\rho=0$ we obtain two identities on $S$
\begin{equation}
    \rd a=K\lrcorner \bar{f},\quad \rd \bar{f}=0.
\end{equation}
we can now compute $Y\lrcorner F$
\begin{equation}
    Y\lrcorner F=\rd \left(\frac{1}{2}A\rho u^2 a-au\right).
\end{equation}
This show invariance, as $\Lie_YF=Y\lrcorner \rd F+\rd\left(Y\lrcorner F\right)=0$.

We showed that the NHG solutions to Einstein-Maxwell equations with cosmological constant $\Lambda$ are invariant under $U,B,Y$ and $K^i\frac{\partial}{\partial x^i}$ vector fields. This finishes proof of Proposition \ref{prop:enhancement}.

\section{Summary and outlook}

We showed that intrinsic rigidity can be understood in terms of properties of foliations by non-expanding horizons of the NHG spacetime. This geometric interpretation also simplifies the original proof and provides a method of showing inheritance of the additional symmetry by the fields. The identity of \cite{Dunajski:2023IntrinsicRigidity, Colling:2025Symmetries} can be obtained by considering the Raychaudhuri equation. This identity is quadratic in $f$ and can thus presumably also be obtained in any spacetime containing an extremal Killing horizon by considering the second variation of the expansion with respect to transversal deformations. One can ask if higher variations provide any valuable insight into the geometry of extremal horizons, especially regarding higher-order deformations \cite{Li:2015wsa, Li:2013gca, Kolanowski:2019wua, Katona:2023vtq, Katona:2024oah}.

In this work, we considered only extremal Killing horizons with trivial topology. It is an interesting question whether intrinsic rigidity also holds for extremal Killing horizons that lack a trivial bundle structure. Non-extremal Killing horizons of non-trivial topology appear as Cauchy horizons in Taub-NUT spacetimes. The role of the extremal case is unclear; nonetheless, these objects are interesting in their own right. Axisymmetric electro-vacuum solutions in four spacetime dimensions are classified \cite{Buk:2025eat, PhysRevD.105.064016}, and no non-axisymmetric solution is known. When the horizon has a non-trivial bundle structure, the constraint equation is exactly the same as in the trivial bundle case \cite{Buk:2025eat}; however, the interpretation of the null convergence condition in terms of horizon data is unknown. Thus, extending our argument to this case requires additional work. The situation is dramatically different if the horizon does not carry a bundle structure with the null vector field generating the fibers. In the generic case, there exists an additional symmetry vector field, but its relation to the Dunajski-Lucietti vector field is unknown.

We expect that the method of proving symmetry inheritance of matter fields can be extended from Maxwell fields to the matter fields considered in \cite{Colling:2025Symmetries}. These matter theories (Yang-Mills fields, $p$-forms with field equations containing specific topological terms) behave nicely at an extremal Killing horizons. We conjecture that for these particular matter contents, there exists an extension of Proposition \ref{prop:enhancement}, although the precise formulation of invariance, especially in the context of Yang-Mills fields, is currently unclear.

We leave these questions for future research.

\subsection*{Declarations}

The work was funded by the grant \emph{Horyzonty i promieniowanie grawitacyjne} No. 2021/43/B/ST2/02950 of the Polish Science Foundation (NCN). The authors declare no conflict of interest. The work is purely theoretical and no shearable data was produced.


\end{document}